\documentclass[11pt]{article}
\usepackage[T1]{fontenc}
\usepackage[utf8]{inputenc}
\usepackage{lmodern,amsmath,amssymb,amsthm,mathtools,booktabs}
\usepackage[margin=1in]{geometry}
\usepackage{microtype}
\usepackage[hidelinks]{hyperref}
\hypersetup{pdftitle={Bivariate Bicycle Codes over Group Algebras},pdfsubject={Weighted-shift BBGA codes, nonabelian realizations, and inequivalence to ordinary 2BGA codes}}
\newtheorem{theorem}{Theorem}[section]
\newtheorem{proposition}[theorem]{Proposition}
\newtheorem{lemma}[theorem]{Lemma}
\theoremstyle{definition}
\newtheorem{definition}[theorem]{Definition}
\theoremstyle{remark}
\newtheorem{remark}[theorem]{Remark}
\newcommand{\F}{\mathbb F_2}
\newcommand{\Z}{\mathbb Z}
\DeclareMathOperator{\wt}{wt}
\DeclareMathOperator{\rk}{rank}
\DeclareMathOperator{\row}{row}
\DeclareMathOperator{\Aut}{Aut}
\DeclareMathOperator{\supp}{supp}
\newcommand{\BBGA}{\operatorname{BBGA}}

\begin{document}
\begin{footnotesize}
\begin{center}
 {\bf Bivariate Bicycle Codes over Group Algebras}\\
{Chaobin Liu}\footnote{Department of Natural Sciences $\&$ Mathematics, Bowie State University, MD, USA\\ Email: cbliu2000@yahoo.com}
\end{center}
\begin{abstract}
We study a weighted-shift formulation of bivariate bicycle group-algebra
(BBGA) codes over finite group algebras. One binary shift acts by right
multiplication and the other by left multiplication; associativity ensures
that these shifts commute even for a nonabelian seed group. Bivariate
polynomials in the shifts, including mixed monomials, therefore define CSS
codes. We give explicit sparsity bounds and distinguish nonabelian seed
groups from codes that are intrinsically inequivalent to abelian
realizations. For $D_3$, a pair of weighted permutation shifts generates a
regular $C_6\times C_{12}$ action and recovers the known $[[144,12,12]]$
and $[[144,14,14]]$ bivariate bicycle codes. We also present a
computer-verified $[[144,16,12]]$ code over $A_4\times C_6$, with every
check of weight eight. This example belongs to the established nonabelian
two-block group-algebra family. Exhaustive low-weight enumeration and an
intrinsic stabilizer-support invariant certify that it is inequivalent to
any ordinary abelian BB or two-block group-algebra realization, even under
qubit permutations combined with local Clifford operations. Finally, a
nonuniform weighted-shift construction over $\mathbb{F}_2[C_3]$ yields a
$[[18,4,3]]$ code with check weights six and eight. An exact enumeration
of its minimum stabilizer supports rules out every ordinary binary 2BGA
realization at that length, under the same equivalence. This establishes
that the unrestricted weighted-shift formulation extends beyond the
ordinary 2BGA class, even with an abelian seed group.
\end{abstract}

\section{Introduction}
Quantum low-density parity-check (qLDPC) codes combine sparse stabilizer
checks with favorable rate and distance scaling. Product constructions
provide several important families; see Ref.~\cite{BE2021} for a review.
At finite blocklength, bivariate bicycle (BB) codes offer compact algebraic
descriptions and sparse checks. The $[[144,12,12]]$ gross code is a prominent
example with weight-six checks~\cite{BCG2024}. Covering-graph constructions
also yield weight-eight BB codes, including a $[[144,14,14]]$
example~\cite{SRB2025}.

For commuting binary square matrices $A$ and $B$, the two-block form
\begin{equation}\label{eq:css}
 H_X=\begin{bmatrix}A&B\end{bmatrix},\qquad
 H_Z=\begin{bmatrix}B^{\mathsf T}&A^{\mathsf T}\end{bmatrix}
\end{equation}
satisfies $H_XH_Z^{\mathsf T}=AB+BA=0$. In BB codes, $A$ and $B$ are
polynomials in commuting cyclic translations. Group-algebra constructions
extend this mechanism through commuting left and right multiplication
operators. This mechanism, including its use for nonabelian groups, is
already part of the theory of two-block group-algebra (2BGA)
codes~\cite{WLP2023,LP2024}.

Here we formulate weighted cyclic shifts over $\F[G]$ and evaluate
bivariate polynomials in their commuting binary lifts. We use the term
\emph{BBGA} for this formulation. When $p=q=1$, $A=\rho(a)$ and
$B=\lambda(b)$, it specializes to an ordinary 2BGA code. Thus neither the
left--right commutation identity nor this specialization is claimed as a
new code family. The formulation permits sparse sums as shift labels and
makes explicit the role of mixed monomials and permutation actions.

Our main distinction is between a nonabelian \emph{description} and
inequivalence of the resulting \emph{code} to an abelian realization.
Group-element labels give commuting permutation shifts, which generate an
abelian permutation group. We demonstrate this reduction for $D_3$ and
recover two known BB codes. We then give an $A_4\times C_6$ example with
parameters $[[144,16,12]]$ and weight-eight checks. A computer-assisted
argument using all minimum-weight stabilizer supports rules out every
ordinary abelian 2BGA realization. This does not assert novelty relative
to the entire nonabelian 2BGA literature or an advantage in decoding or
syndrome-extraction performance.

A further distinction concerns the scope of the weighted-shift formulation
relative to the entire ordinary 2BGA class. In
Section~\ref{sec:weighted18}, we give a $[[18,4,3]]$ example with all seeds
nonzero, check weights at most eight, and minimum stabilizer weight six.
Its 38 minimum stabilizer supports violate a divisibility condition imposed
by every possible ordinary 2BGA realization of length 18. Unlike the
$A_4\times C_6$ example, which remains an ordinary nonabelian 2BGA code,
this weighted-shift example lies outside the ordinary 2BGA class altogether.
The result establishes a strict extension up to qubit permutations and
local Clifford operations, rather than a claim of optimal code parameters.

The remainder of the paper is organized as follows.
Section~\ref{sec:construction} defines the BBGA construction and establishes
its commutation, CSS orthogonality, and sparsity properties.
Section~\ref{sec:abelian} analyzes commuting permutation shifts and presents
the $D_3$ realizations of known BB codes.
Section~\ref{sec:nonabelian} gives the $[[144,16,12]]$ nonabelian example
and its certificate of inequivalence to ordinary abelian 2BGA codes.
Section~\ref{sec:weighted18} presents the $[[18,4,3]]$ weighted-shift code
and proves its inequivalence to every ordinary binary 2BGA realization.
Section~\ref{sec:conclusion} summarizes the results and outlines further
research directions. Appendix~\ref{app:reproduce} specifies the coordinate
conventions and provides instructions for reproducing the computations.

\section{The BBGA construction}\label{sec:construction}
\subsection{Group algebra and multiplication conventions}
Let $G$ be a finite group of order $\ell$, with identity $e$. Its binary
group algebra is
\[
 \mathcal R=\F[G]=\left\{\sum_{h\in G}z_hh:z_h\in\F\right\}.
\]
Addition is coefficientwise and multiplication is extended bilinearly from
$G$. Fix a group ordering and identify $\mathcal R$ with column vectors in
$\F^\ell$. Define binary matrices $\lambda(z)$ and $\rho(z)$ by
\begin{equation}\label{eq:regular}
 \lambda(z)e_h=\sum_{g\in G}z_ge_{gh},\qquad
 \rho(z)e_h=\sum_{g\in G}z_ge_{hg}.
\end{equation}
Thus $\lambda(z)u=zu$ and $\rho(z)u=uz$. With this column-vector
convention,
\[
 \lambda(z)\lambda(w)=\lambda(zw),\qquad
 \rho(z)\rho(w)=\rho(wz).
\]
In particular, right multiplication is an antihomomorphism, or a
representation of the opposite algebra. Associativity gives
\begin{equation}\label{eq:LR}
 \lambda(z)\rho(w)=\rho(w)\lambda(z)\qquad(z,w\in\mathcal R),
\end{equation}
since $z(uw)=(zu)w$. This identity does not require $G$ to be abelian.

\subsection{Weighted shifts and binary lifts}
For $m\geq1$, let $S_m$ be the left cyclic shift with
$(S_m)_{i,j}=1$ exactly when $j=i+1\pmod m$, using indices in $\Z_m$.
For $z_0,\ldots,z_{m-1}\in\mathcal R$, let
$S_m(z_0,\ldots,z_{m-1})$ have entry $z_i$ at $(i,i+1)$ and zero elsewhere.
Choose $p,q\geq1$ and seeds
\[
 \boldsymbol a=(a_0,\ldots,a_{p-1}),\qquad
 \boldsymbol b=(b_0,\ldots,b_{q-1})
\]
in $\mathcal R$, and set
\[
 \widehat P=S_p(\boldsymbol a)\otimes I_q,\qquad
 \widehat Q=I_p\otimes S_q(\boldsymbol b).
\]
These matrices lie in $M_{pq}(\mathcal R)$. They commute when $G$ is
abelian, but need not commute otherwise.

Set $N=pq\ell$. Replace the entries of $\widehat P$ by their right
multiplication matrices and those of $\widehat Q$ by their left
multiplication matrices. Explicitly, the $\ell\times\ell$ blocks are
\begin{align}
 P_{(i,j),(i',j')}&=\delta_{i',i+1}\delta_{j',j}\,\rho(a_i),\label{eq:P}\\
 Q_{(i,j),(i',j')}&=\delta_{i',i}\delta_{j',j+1}\,\lambda(b_j),\label{eq:Q}
\end{align}
where the first and second coordinates are reduced modulo $p$ and $q$,
respectively. Both $P$ and $Q$ are binary $N\times N$ matrices.

\begin{proposition}\label{prop:commute}
For arbitrary group-algebra seeds, $PQ=QP$.
\end{proposition}
\begin{proof}
The only potentially nonzero block in row $(i,j)$ of either product is
in column $(i+1,j+1)$. In $PQ$ this block is
$\rho(a_i)\lambda(b_j)$; in $QP$ it is
$\lambda(b_j)\rho(a_i)$. They agree by~\eqref{eq:LR}.
\end{proof}

\subsection{Polynomial matrices and CSS parameters}
Let $\mathcal S_A,\mathcal S_B\subset\Z_{\geq0}^2$ be finite sets and define
\begin{equation}\label{eq:polynomials}
 A=f(P,Q)=\sum_{(u,v)\in\mathcal S_A}P^uQ^v,\qquad
 B=g(P,Q)=\sum_{(u,v)\in\mathcal S_B}P^uQ^v.
\end{equation}
All sums and matrix operations are over $\F$. A monomial with $u,v>0$
is called mixed. Restricting to terms with $u=0$ or $v=0$ gives the
pure-power case. If both shifts are invertible, Laurent polynomials with
integer exponents are also allowed.

\begin{definition}
The BBGA code $\BBGA_{G,p,q}(\boldsymbol a,\boldsymbol b;f,g)$ is the
CSS code defined by~\eqref{eq:css} and~\eqref{eq:polynomials}. The first
argument of $f$ and $g$ is always evaluated at $P$, and the second at $Q$.
\end{definition}

\begin{theorem}
The BBGA construction defines a valid CSS code of length $n=2pq|G|$.
\end{theorem}
\begin{proof}
By Proposition~\ref{prop:commute}, every two monomials in $P,Q$ commute.
Hence $AB=BA$, and~\eqref{eq:css} gives
$H_XH_Z^{\mathsf T}=AB+BA=0$ in characteristic two.
\end{proof}

Its dimension is
\begin{equation}\label{eq:dimension}
 k=2pq|G|-\rk_{\F}H_X-\rk_{\F}H_Z.
\end{equation}
For $k>0$, identify row vectors with their transposes when comparing row
spaces and kernels, and define
\begin{align}
 d_X&=\min\{\wt(v):v\in\ker H_Z\setminus\row H_X\},\label{eq:dx}\\
 d_Z&=\min\{\wt(v):v\in\ker H_X\setminus\row H_Z\}.\label{eq:dz}
\end{align}
Then $d=\min(d_X,d_Z)$. Excluding the stabilizer row space is essential:
low-weight vectors in a parity-check kernel may be stabilizers rather
than logical operators.

\subsection{Sparsity}
For $z\in\mathcal R$, write $|z|=|\{g:z_g=1\}|$, and set
$\alpha=\max_i|a_i|$ and $\beta=\max_j|b_j|$. Every row and column of $P$
has weight at most $\alpha$, and every row and column of $Q$ has weight at
most $\beta$. The maximum row and column weights of a product are bounded
by the products of the corresponding maxima. Consequently, for
nonnegative exponents, define
\[
 w_A=\sum_{(u,v)\in\mathcal S_A}\alpha^u\beta^v,\qquad
 w_B=\sum_{(u,v)\in\mathcal S_B}\alpha^u\beta^v,
\]
with zeroth powers interpreted as one. Every check has weight at most
$w_A+w_B$, and each qubit participates in at most $w_A+w_B$ checks of
each type. Cancellations can lower these bounds. A family is LDPC if the
row and column weights remain uniformly bounded; CSS orthogonality alone
does not imply this property. Inverses of sparse matrices may be dense,
so Laurent-polynomial instances require a separate sparsity check.

If every seed is a single group element, all monomials are permutation
matrices. The check-weight bound then simplifies to
$|\mathcal S_A|+|\mathcal S_B|$. Equality holds when the monomials within
each defining matrix have disjoint supports. Distinct permutation matrices
need not have disjoint supports in general; distinct translations in a
regular action do.

\section{Permutation shifts and abelian realizations}\label{sec:abelian}
\begin{proposition}\label{prop:orbits}
Suppose all seeds are group elements. The permutation group
$K=\langle P,Q\rangle$ is abelian. On each of its coordinate orbits, the
induced faithful action is regular and abelian. If $K$ is transitive, every
BBGA code built from these shifts is permutation-equivalent to an ordinary
abelian 2BGA code.
\end{proposition}
\begin{proof}
The generators are commuting permutations. On an orbit $\Omega$, a
stabilizer $K_\omega$ fixes every point: if $k\omega=\omega$, then
$k(h\omega)=h(k\omega)=h\omega$. Thus the faithful quotient
$K/K_\omega$ acts regularly on $\Omega$. In the transitive case, label
coordinates by this abelian quotient. Both shifts, and hence $A,B$, become
linear combinations of its regular translations. Applying the same
relabeling to both qubit blocks gives the claimed realization.
\end{proof}

For multiple orbits, the simultaneous orbit decomposition makes $A$ and
$B$ block diagonal and decomposes the CSS code into components with
abelian regular actions. This statement does not require the components
to be realizations over one common abelian group.

\subsection{A regular action from the nonabelian group \texorpdfstring{$D_3$}{D3}}
Let
\[
 D_3=\langle r,s\mid r^3=s^2=e,\ srs=r^{-1}\rangle,
 \qquad p=3,\quad q=4,
\]
and choose
\begin{equation}\label{eq:D3seeds}
 \boldsymbol a=(e,e,s),\qquad \boldsymbol b=(e,e,e,r).
\end{equation}
The binary shifts have size $72\times72$. Although the corresponding
algebra-valued shifts fail to commute, the binary lifts commute.

\begin{lemma}\label{lem:D3}
The shifts in~\eqref{eq:D3seeds} generate a regular action
$\langle P,Q\rangle\cong C_6\times C_{12}$ on the 72 coordinates.
\end{lemma}
\begin{proof}
In the ordering $(i,j,h)\in\Z_3\times\Z_4\times D_3$,
\[
 P^3=I_{12}\otimes\rho(s),\qquad
 Q^4=I_{12}\otimes\lambda(r).
\]
The cyclic coordinate shifts force the orders of $P$ and $Q$ to be
multiples of three and four, respectively. The displayed identities then
give orders six and twelve.

If $P^uQ^v$ fixes a coordinate, then $u=3a$ and $v=4b$. Its action on
the group coordinate is $h\mapsto r^bhs^a$. A fixed point would imply
$r^b=hs^{-a}h^{-1}$. The two sides lie in subgroups of orders three and
two, respectively, so both are the identity. Thus $u=0\pmod6$ and
$v=0\pmod{12}$. The $6\cdot12=72$ products act freely on 72 coordinates
and hence regularly.
\end{proof}

For comparison with standard BB notation, set $x=Q$ and $y=P$ in the
following displayed formulas; then $x^{12}=y^6=I$. This is a notational
swap, not a change to the argument convention in the BBGA definition.

\subsection{Two known BB codes}
The gross code~\cite{BCG2024} has
\begin{equation}\label{eq:gross}
 A=Q^3+P+P^2,\qquad B=P^3+Q+Q^2.
\end{equation}
Each matrix is the sum of three distinct regular translations, so every
check has weight six. Binary elimination gives
$\rk H_X=\rk H_Z=66$ and hence $k=12$. By Lemma~\ref{lem:D3}, the
code is a coordinate relabeling of the standard gross code, whose
parameters are $[[144,12,12]]$.

The weight-eight example of Symons, Rajput, and Browne~\cite{SRB2025}
is obtained from
\begin{align}
 A&=Q^6P^4+Q^5P^4+Q^3+Q^{11}P^3,\label{eq:coverA}\\
 B&=P^5+Q^8P+Q^5P^5+Q^9P^4.\label{eq:coverB}
\end{align}
Each defining matrix is a sum of four distinct regular translations.
The ranks are both 65, giving $k=14$. These are precisely the published
polynomials on $C_{12}\times C_6$ with $x=Q,y=P$, so the reported distance
14 transfers under the coordinate relabeling. The code therefore has
parameters $[[144,14,14]]$. The distances of these two examples are inherited
from the cited BB codes; their dimensions and check weights are also
verified in the accompanying scripts.

\subsection{A further weight-eight instance}\label{sec:D3candidate}
The same shifts give
\begin{equation}\label{eq:D3third}
 A=P^4+Q^2+Q^3+PQ,\qquad B=P^2+P^3+Q+(PQ)^5.
\end{equation}
Both check matrices have rank 64, so $k=16$, and every check has weight
eight. For an explicit upper bound on distance, order
$D_3=(e,r,r^2,s,rs,r^2s)$ and then order $(i,j,h)$ lexicographically.
Number the full 144 qubits from \emph{one}, with the first block first.
The supports
\begin{align*}
 U_X&=\{4,5,6,7,8,9,16,17,18,19,20,21\},\\
 U_Z&=\{1,4,15,18,27,30,38,41,50,53,61,64\}
\end{align*}
define nontrivial $X$- and $Z$-logical operators, respectively. Thus
$d_X,d_Z\leq12$. The supplement verifies their zero syndromes and their
nonmembership in the corresponding stabilizer spaces. No exhaustive
lower-bound certificate for this instance is included, so we retain only
the verified statement $[[144,16,d]]$ with $d\leq12$. In particular, the
weight-twelve witnesses alone do not prove $d=12$. By Lemma~\ref{lem:D3},
this instance is an abelian realization regardless of its exact distance.

\section{A certified nonabelian example}\label{sec:nonabelian}
We now use sparse group-algebra sums rather than group-element labels.
Here $p=q=1$, so the example is a nonabelian 2BGA code within the BBGA
formulation~\cite{LP2024}.

\subsection{Construction and parameters}
Let $G=A_4\times C_6$, where
\[
 x=(1\ 2\ 3),\qquad y=(1\ 2)(3\ 4)
\]
generate $A_4$, permutation products act from right to left, and $t$ is a
central element of order six. Thus
$x^3=y^2=(xy)^3=t^6=e$, $tx=xt$, and $ty=yt$. In $\F[G]$, set
\begin{align}
 a&=t^5+xt+(xyx)t^4+(yxy)t^2,\label{eq:nonaba}\\
 b&=e+(xyx^2)t^5+(yxy)t^4+(yx^2)t.\label{eq:nonabb}
\end{align}
Take $P=\rho(a)$, $Q=\lambda(b)$, $A=P$, and $B=Q$. Each seed has four
distinct terms. Therefore every row and column of $A,B$ has weight four,
and every row of $H_X,H_Z$ has weight eight. Each qubit belongs to four
$X$ checks and four $Z$ checks.

\begin{proposition}[Computer-assisted parameter certificate]\label{prop:parameters}
The code defined by~\eqref{eq:nonaba}--\eqref{eq:nonabb} has
$\rk H_X=\rk H_Z=64$ and $d_X=d_Z=12$. Its parameters are
$[[144,16,12]]$.
\end{proposition}
\begin{proof}
The construction has $n=2|G|=144$. Exact binary elimination gives the two
ranks and verifies that the all-ones vector belongs to both check row
spaces. Consequently, both parity-check kernels contain only even-weight
vectors. The logical supports in Appendix~\ref{app:reproduce} establish
$d_X,d_Z\leq12$.

For the lower bound, every even-weight support of size at most ten can be
written as the symmetric difference of two five-element supports: split
it equally and pad both halves with the same coordinates outside the
support. Thus it suffices to compare the syndromes of all five-element
supports and show that equal syndromes imply equal residues modulo the
opposite stabilizer row space.

The exhaustive computation uses the free simultaneous central $C_6$
translation. Every five-element support has an orbit of size six, because
an invariant support for a nonidentity translation would have cardinality
divisible by two or three. The number of support orbits is therefore
\[
 \binom{144}{5}/6=80\,168\,088.
\]
The verifier takes one representative from each orbit, minimizes its
syndrome over all six translations, and retains every translation attaining
that minimum. It sorts exact 72-bit syndromes and compares exact
144-bit stabilizer residues within each equal-syndrome class. Both sectors
have $80\,168\,088$ retained entries and $57\,120$ repeated-key entries;
every repeated key has a stabilizer difference. Because the translation
preserves both relevant spaces, this checks all pairs before symmetry
reduction as well. Hence neither sector has a nontrivial logical vector of
weight at most ten. Odd weights are already excluded, proving the lower
bound twelve. The source code and completed output are included in the
supplement.
\end{proof}

\subsection{An intrinsic certificate of inequivalence}
A small automorphism group of a selected Tanner graph does not by itself
rule out code equivalence: another stabilizer generating set could have a
different graph. We instead use the supports of \emph{all} minimum-weight
stabilizers, disregarding their Pauli labels and phases.

\begin{lemma}[Computer-assisted minimum-support classification]\label{lem:minsupports}
For the code in Proposition~\ref{prop:parameters}, the minimum nonidentity
stabilizer weight is eight. The complete set of minimum stabilizer supports
consists of the 72 rows of $H_X$ and the 72 rows of $H_Z$, with no repeated
support between the two collections.
\end{lemma}
\begin{proof}
Exhaustive syndrome comparison of supports of size at most four enumerates
all kernel vectors of weight at most eight. The number of supports compared,
including the empty support, is
\[
 \sum_{j=0}^4\binom{144}{j}=17\,676\,661.
\]
In $\ker H_Z$, the computation finds exactly 72 nonzero such vectors, all
of weight eight and exactly the rows of $H_X$. Interchanging $X$ and $Z$
gives the analogous statement for $\ker H_X$. The two support collections
are disjoint.

Every Pauli stabilizer has an $X$ component in $\row H_X$ and a $Z$
component in $\row H_Z$, and its support is the union of their supports.
If both components are nonzero and the union has size at most eight, the
classification forces both supports to be identical weight-eight supports.
Their disjointness excludes this possibility. The pure components give
exactly the displayed 144 minimum supports.
\end{proof}

Form a bipartite incidence graph $\Gamma$ with 144 qubit vertices and 144
vertices for the supports in Lemma~\ref{lem:minsupports}. Join a support
to each qubit it contains. Color only the two vertex classes; do not give
$X$ and $Z$ supports different colors. An exact nauty
computation~\cite{MP2014} gives
\begin{equation}\label{eq:aut}
 |\Aut(\Gamma)|=12,
\end{equation}
with twelve qubit orbits of size twelve. The graph and automorphism
generators are included in the supplement.

\begin{theorem}[Inequivalence to ordinary abelian 2BGA codes]\label{thm:inequivalence}
The code in Proposition~\ref{prop:parameters} is inequivalent to every
ordinary abelian BB or 2BGA code of length 144 under qubit permutations,
even when combined with single-qubit Clifford operations.
\end{theorem}
\begin{proof}
An ordinary abelian 2BGA realization of length 144 uses the regular
representation of an abelian group $H$ of order 72 on each of its two
qubit blocks. Simultaneous translations by $H$ act faithfully on the
qubits and preserve both stabilizer spaces. They therefore preserve the
complete set of minimum stabilizer supports. Its intrinsic support graph
must contain a subgroup of automorphisms of order 72.

Qubit permutations relabel Pauli supports, and single-qubit Clifford
operations preserve them. An equivalence of the stated type would thus
induce an isomorphism of the intrinsic support graphs and transport this
order-72 subgroup into $\Aut(\Gamma)$. This contradicts~\eqref{eq:aut},
since 72 does not divide 12.
\end{proof}

The theorem concerns the standard square two-block construction from one
abelian group. It makes no assertion about arbitrary rectangular
lifted-product constructions or equivalence under general entangling
Clifford transformations.

\subsection{Parameter comparison}
Table~\ref{tab:comparison} summarizes the certified length-144 examples. The
nonabelian code encodes four more qubits than the gross code at the same
length and distance, but increases the check weight from six to eight. It
does not dominate $[[144,14,14]]$, whose distance is larger. Its parameters
and check weight match a reported double-layer reflection
code~\cite{LXX2026}; equality of these quantities does not establish
code equivalence. No threshold, decoding, or circuit-depth advantage is
inferred from the parameter comparison.

\begin{table}[htbp]
\centering
\caption{Length-144 examples. $W_X,W_Z$ denote check weights in the
specified presentations. The uncertified distance instance
in Section~\ref{sec:D3candidate} is omitted.}\label{tab:comparison}
\begin{tabular}{lcccl}
\toprule
Realization & $k$ & $d$ & $(W_X,W_Z)$ & Abelian status\\
\midrule
$D_3$, Eq.~\eqref{eq:gross} &12&12&$(6,6)$&Known abelian BB\\
$D_3$, Eqs.~\eqref{eq:coverA}--\eqref{eq:coverB}&14&14&$(8,8)$&Known abelian BB\\
$A_4\times C_6$, Eqs.~\eqref{eq:nonaba}--\eqref{eq:nonabb}&16&12&$(8,8)$&Excluded by Theorem~\ref{thm:inequivalence}\\
\bottomrule
\end{tabular}
\end{table}

\section{A weighted-shift code beyond ordinary 2BGA realizations}
\label{sec:weighted18}
The preceding example separates nonabelian 2BGA codes from the ordinary
abelian subclass. We now establish a stronger separation: a nontrivial
weighted-shift BBGA code that admits no ordinary binary 2BGA realization
at all. Here an ordinary 2BGA realization means the full square two-block
construction over one finite group algebra, with length twice the group
order~\cite{LP2024}.

\subsection{Explicit weighted shifts}
Let $G=C_3=\langle r\mid r^3=e\rangle$, take $p=3$, $q=1$, and set
$c=r+r^2\in\F[G]$. Choose
\begin{equation}\label{eq:weighted18seeds}
 \boldsymbol a=(e,e,c),\qquad \boldsymbol b=(c).
\end{equation}
All seeds are nonzero, although $c$ is not a unit. In the group ordering
$(e,r,r^2)$, define
\[
 C=\begin{pmatrix}0&1&1\\1&0&1\\1&1&0\end{pmatrix},\qquad
 I=I_3,\qquad J=I+C.
\]
Both multiplication matrices of $c$ equal $C$, and $C^2=C$. The binary
weighted shifts are therefore
\begin{equation}\label{eq:weighted18PQ}
 P=\begin{pmatrix}0&I&0\\0&0&I\\C&0&0\end{pmatrix},\qquad
 Q=\begin{pmatrix}C&0&0\\0&C&0\\0&0&C\end{pmatrix}.
\end{equation}
They satisfy $P^3=Q$, $Q^2=Q$, and $PQ=QP$. Take
\begin{equation}\label{eq:weighted18poly}
 A=P^2+Q,\qquad B=P^2+PQ+P^2Q.
\end{equation}
Equivalently, the defining matrices are
\begin{equation}\label{eq:weighted18AB}
 A=\begin{pmatrix}C&0&I\\C&C&0\\0&C&C\end{pmatrix},\qquad
 B=\begin{pmatrix}0&C&J\\0&0&C\\C&0&0\end{pmatrix}.
\end{equation}
Equation~\eqref{eq:css} then gives the check matrices. Their orthogonality
follows from $AB=BA$.

\subsection{Parameters and minimum stabilizer supports}
\begin{proposition}[Exact finite verification]\label{prop:weighted18}
The code defined by~\eqref{eq:weighted18seeds}--\eqref{eq:weighted18poly}
has parameters $[[18,4,3]]$. Each check matrix has six rows of weight six
and three rows of weight eight. Its minimum nonidentity stabilizer weight
is six, and there are exactly 38 distinct minimum stabilizer supports.
\end{proposition}
\begin{proof}
Binary elimination gives $\rk H_X=\rk H_Z=7$, so $k=18-7-7=4$.
The row weights follow directly from~\eqref{eq:weighted18AB}. Both check
matrices have 18 distinct nonzero columns. Thus neither parity-check kernel
contains a nonzero vector of weight one or two.

For explicit logical witnesses, order the coordinates within each
nine-qubit block as $(i,r^j)$, with $i$ first and $j$ second,
$0\leq i,j<3$. Number all qubits from one, with the first block first.
The operators
\[
 X_4X_5X_6,\qquad Z_1Z_2Z_3
\]
have zero syndromes of their respective types and lie outside the
corresponding stabilizer spaces. Hence $d_X=d_Z=3$. These are separate
distance witnesses, not an anticommuting logical pair.

For each of the two binary stabilizer row spaces, enumeration of its
$2^7=128$ vectors gives the weight enumerator
\begin{equation}\label{eq:weighted18enumerator}
 1+19z^6+45z^8+42z^{10}+18z^{12}+3z^{14}.
\end{equation}
The 19 weight-six supports in the $X$ space and the 19 in the $Z$ space
are disjoint collections. A mixed Pauli stabilizer has support equal to
the union of its component supports. Such a union cannot have size six
unless the two nonzero components have identical weight-six supports,
which the disjointness excludes. Thus the 38 pure supports constitute
all minimum stabilizer supports, and there is no nonidentity stabilizer
of smaller weight. Enumeration of all $2^{14}=16\,384$ stabilizers
independently verifies this classification. The accompanying standard-library
Python program performs these exact finite computations.
\end{proof}

The combined Tanner graph of the displayed $X$ and $Z$ checks is connected,
as also checked by the verifier. In particular, this presentation has no
unused qubits. Connectedness of this presentation alone is not used as an
obstruction to code equivalence.

\subsection{Exclusion of every ordinary binary 2BGA realization}
\begin{theorem}\label{thm:weighted18}
The code in Proposition~\ref{prop:weighted18} is inequivalent to every
ordinary binary 2BGA code of length 18 under qubit permutations combined
with single-qubit Clifford operations.
\end{theorem}
\begin{proof}
An ordinary length-18 2BGA code uses a group $H$ of order nine. Every group
of order nine is abelian: $H\cong C_9$ or $C_3\times C_3$. Simultaneous
regular translations on the two nine-qubit blocks therefore preserve the
stabilizer code~\cite{LP2024}. This action is free on the qubits.

Consider a six-qubit support $S$, and let $H_S$ be its setwise stabilizer
under translation. The restricted action of $H_S$ on $S$ is free, so
$|H_S|$ divides $|S|=6$. It also divides $|H|=9$, hence
$|H_S|\in\{1,3\}$. Every orbit of six-qubit supports therefore has size
$9/|H_S|\in\{9,3\}$. Any translation-invariant collection of such supports
must have cardinality divisible by three.

The complete collection of minimum stabilizer supports is intrinsic to
the code. Qubit permutations relabel this collection, and local Clifford
operations preserve Pauli supports. An equivalence of the stated type
would consequently transport the free translation action to the
38 minimum supports of Proposition~\ref{prop:weighted18}. Since three
does not divide 38, this is impossible.
\end{proof}

There is also a qubit-incidence form of the obstruction. Let $\nu_i$
count the minimum stabilizer supports containing qubit $i$. Enumeration
gives $\nu_i=12$ on $\{4,5,6,13,14,15\}$ and $\nu_i=13$ on the other
twelve qubits. Thus there is an intrinsic six-element set of qubits,
whereas every invariant qubit set under a free order-nine translation
action must have cardinality divisible by nine.

\begin{remark}
The seed group in this example is abelian. Its additional expressive power
comes from the nonuniform singular shift labels, not from group
noncommutativity. Since every ordinary 2BGA code is included in the BBGA
formulation at $p=q=1$, Theorem~\ref{thm:weighted18} establishes strict
containment of the ordinary binary 2BGA class in the unrestricted
weighted-shift BBGA class, up to the specified equivalence. The conclusion
does not exclude general matrix-valued lifted-product descriptions,
constructions with added or removed qubits, or equivalence under arbitrary
entangling Clifford circuits. It makes no claim of literature novelty or
parameter optimality.
\end{remark}

\section{Conclusion}\label{sec:conclusion}
The weighted-shift BBGA formulation uses the established commutation of
left and right multiplication to produce CSS codes from bivariate
polynomials over arbitrary finite group algebras. Explicit weight bounds
identify when the resulting construction remains sparse. Group-element
seeds lead to abelian permutation actions; the $D_3$ examples illustrate
how a nonabelian description can recover ordinary BB codes.

Sparse group-algebra sums also give instances outside the ordinary
abelian 2BGA equivalence class. The $A_4\times C_6$ example has certified
parameters $[[144,16,12]]$, weight-eight checks, and an intrinsic
stabilizer-support obstruction to an abelian realization. It is itself a
member of the existing nonabelian 2BGA family. The useful distinction is
therefore between the seed description, the binary action, and the
stabilizer code up to the specified equivalence.

The $[[18,4,3]]$ code over $\F[C_3]$ establishes a different separation:
nonuniform weighted shifts can yield a code outside every ordinary binary
2BGA realization at the same length. Its check weights are six and eight,
and the 38 minimum stabilizer supports provide a direct obstruction to
the translations required by any order-nine group realization. Thus an
abelian seed algebra need not imply an ordinary abelian 2BGA code when
nonuniform group-algebra-valued shift labels are allowed.

Further work should characterize the weighted-shift instances outside the
ordinary 2BGA class, beyond the explicit separation proved here, and seek
larger examples with stronger rate--distance tradeoffs under fixed
check-weight constraints. Determining the exact distance of
Eq.~\eqref{eq:D3third} remains a separate finite-length question.
Asymptotic scaling, decoding, and fault-tolerant measurement costs also
require further analysis.

\section{Acknowledgments}

In the preparation of this work, the author made use of GPT-5.6 Sol and GPT-6 Astra for grammar and clarity refinement, as well as to aid in computing the code parameters presented in the examples. These interactions helped improve the quality of this paper. After using this tool, the author carefully reviewed and edited the content as necessary and takes full responsibility for the final content of this paper.

\appendix
\section{Reproducibility and coordinate conventions}\label{app:reproduce}
The accompanying supplement contains the reconstruction scripts,
exhaustive enumeration programs, matrix data, and verification logs.
The nonabelian computations use exact binary arithmetic, Python with
NumPy, C++17, and nauty through pynauty (tested version 2.8.8.1).

For the $A_4\times C_6$ code, order the even permutations of
$(0,1,2,3)$ lexicographically by their image tuples. If $h$ is the
$j$th permutation, starting at $j=0$, assign $ht^u$ the index $6j+u$,
where $0\leq u<6$. Number the qubits by
$72b+6j+u$ for block $b\in\{0,1\}$, starting at \emph{zero}.
This convention differs from the one-based $D_3$ witnesses in
Section~\ref{sec:D3candidate}.
The group-algebra support indices are
\[
 \supp(a)=\{5,25,34,56\},\qquad
 \supp(b)=\{0,71,58,61\}.
\]
Nontrivial weight-twelve logical supports are
\begin{align*}
 U_X&=\{10,11,18,43,61,95,97,120,123,124,125,136\},\\
 U_Z&=\{84,91,92,105,111,112,114,119,122,131,136,139\}.
\end{align*}
The reconstruction script checks $U_X\in\ker H_Z\setminus\row H_X$
and $U_Z\in\ker H_X\setminus\row H_Z$ directly.

From the supplement directory, run
\begin{verbatim}
python nonabelian/verify.py
python nonabelian/verify.py --exhaustive
python d3/verify.py
\end{verbatim}
The first command reconstructs the nonabelian matrices and checks
commutation, ranks, weights, logical witnesses, symmetry, and the
support-graph automorphism group. Only the second command independently
recomputes the distance lower bound and the completeness of the minimum
stabilizer-support collection. It requires a C++17 compiler and
approximately 1.4 GB of memory. The third command checks the dimensions
and weights of all three $D_3$ instances and the weight-twelve witnesses
for the third; it does not certify a distance lower bound for that instance.

The SHA-256 digests of the nonabelian check matrices, serialized as
row-major unsigned eight-bit entries, are
\begin{quote}\small
$H_X$:\\
\nolinkurl{cc050e8ec7b466b3b9d8b4ee26e0be0645c4243abe865b4596c20192ffa8d914}\\[3pt]
$H_Z$:\\
\nolinkurl{6c0c76415529174c6a0cb9e139e033e09a5b5044935c6a032acd0a091f44ccd3}.
\end{quote}
The graph certificate is meaningful as a code invariant only together
with the exhaustive minimum-support classification. No distance lower
bound is inferred merely from a randomized search failing to find a
lighter logical operator.

\subsection{The length-18 weighted-shift example}
The directory \texttt{weighted18/} contains a self-contained verifier,
the complete binary matrices and minimum-support list in
\texttt{certificate.json}, and the completed verification output.
From the supplement directory, run
\begin{verbatim}
python weighted18/verify.py
\end{verbatim}
This computation uses only the Python standard library. It checks
commutation, CSS orthogonality, ranks, row weights, distinct nonzero
syndrome columns, and the weight-three logical witnesses. It enumerates
both stabilizer row spaces and all $16\,384$ Pauli stabilizers, verifying
Eq.~\eqref{eq:weighted18enumerator}, the 38 distinct minimum supports,
and the incidence counts used in Section~\ref{sec:weighted18}.
No randomized distance search, external solver, or graph-automorphism
computation is needed. Coordinates are one-based, as specified in the
proof of Proposition~\ref{prop:weighted18}; the option
\texttt{--write-data} regenerates the certificate file.
\end{footnotesize}
\end{document}